\documentclass[11pt]{article}
\usepackage{fullpage}
\usepackage{times}
\usepackage{subfigure}
\usepackage{color}
\usepackage{url}
\usepackage{graphicx}
\usepackage{amsmath}
\usepackage{amssymb}
\newcommand{\qed}{\mbox{}\hspace*{\fill}\nolinebreak\mbox{$\rule{0.6em}{0.6em}$}
}
\newcommand{\expect}{{\bf \mbox{\bf E}}}
\newcommand{\prob}{{\bf \mbox{\bf Pr}}}

\definecolor{gray}{rgb}{0.5,0.5,0.5}

\newtheorem{theorem}{Theorem}
\newtheorem{lemma}[theorem]{Lemma}

\newtheorem{corollary}[theorem]{Corollary}

\newtheorem{remark}[theorem]{Remark}

\newenvironment{proof}{{\bf Proof:}}{$\qed$\par}

\newenvironment{proofof}[1]{\noindent{\bf Proof of #1:}}{$\qed$\par}

\usepackage{hyperref}
\hypersetup{
bookmarksnumbered
}

\newcommand{\M}{\mathbf M}

\begin{document}

\title{\Large Perfect Matchings in $O(n \log n)$ Time in Regular Bipartite Graphs\thanks{A preliminary version of the paper appeared in STOC'10.}}
\author{Ashish Goel\thanks{
    Departments of Management Science and Engineering and (by courtesy)
    Computer Science, Stanford University.
    Email: {\tt ashishg@stanford.edu}.
    Research supported by NSF
    ITR grant 0428868, NSF CAREER award 0339262, and a grant from the
    Stanford-KAUST alliance for academic excellence.}\\
\and Michael Kapralov\thanks{
    Institute for Computational and Mathematical Engineering, Stanford University.
    Email: {\tt kapralov@stanford.edu}. Research supported by a Stanford Graduate Fellowship.}\\   
\and Sanjeev Khanna\thanks{Department of Computer and Information Science, University of Pennsylvania,
Philadelphia PA. Email: {\tt sanjeev@cis.upenn.edu}. Supported in
part by NSF Award CCF-0635084.}
}

\maketitle
\begin{abstract}
  In this paper we consider the well-studied problem of finding a perfect
  matching in a $d$-regular bipartite graph on $2n$ nodes with $m=nd$
  edges. The best-known algorithm for general bipartite graphs (due to
  Hopcroft and Karp) takes time $O(m\sqrt{n})$. In regular bipartite graphs,
  however, a matching is known to be computable in $O(m)$ time (due to Cole,
  Ost and Schirra). In a recent line of work by Goel, Kapralov and Khanna the
  $O(m)$ time algorithm was improved first to $\tilde O\left(\min\{m,
    n^{2.5}/d\}\right)$ and then to $\tilde O\left(\min\{m,
    n^2/d\}\right)$. It was also shown that the latter algorithm is optimal up
  to polylogarithmic factors among all algorithms that use non-adaptive
  uniform sampling to reduce the size of the graph as a first step.
  
  In this paper, we give a randomized algorithm that finds a perfect matching
  in a $d$-regular graph and runs in $O(n\log n)$ time (both in expectation
  and with high probability). The algorithm performs an appropriately
  truncated random walk on a modified graph to successively find augmenting
  paths. Our algorithm may be viewed as using adaptive uniform sampling, and
  is thus able to bypass the limitations of (non-adaptive) uniform sampling
  established in earlier work. We also show that randomization is crucial for
  obtaining $o(nd)$ time algorithms by establishing an $\Omega(nd)$ lower
  bound for any deterministic algorithm.  Our techniques also give an
  algorithm that successively finds a matching in the support of a doubly
  stochastic matrix in expected time $O(n\log^2 n)$ time, with $O(m)$
  pre-processing time; this gives a simple $O(m+mn\log^2 n)$ time algorithm
  for finding the Birkhoff-von Neumann decomposition of a doubly stochastic
  matrix.
\end{abstract}

\section{Introduction}


A bipartite graph $G = (P, Q, E)$ with vertex set $P \cup Q$ and edge set $E
\subseteq P \times Q$ is said to be $d$-regular if every vertex has the same
degree $d$. We use $m=nd$ to denote the number of edges in $G$ and $n$ to
represent the number of vertices in $P$ (as a consequence of regularity, $P$
and $Q$ have the same size). Regular bipartite graphs have been studied
extensively, in particular in the context of expander constructions,
scheduling, routing in switch fabrics, and task-assignment
~\cite{mr:random,amsz:color2003,cos:regular2001}.

A regular bipartite graph of degree $d$ can be decomposed into exactly $d$
perfect matchings, a fact that is an easy consequence of Hall's
theorem~\cite{b:graphtheory} and is closely related to the Birkhoff-von
Neumann decomposition of a doubly stochastic matrix~\cite{b:bvn46,vn:bvn53}.
Finding a matching in a regular bipartite graph is a well-studied problem,
starting with the algorithm of K\"{o}nig in 1916~\cite{k:regular16}, which is
now known to run in time $O(mn)$. The well-known bipartite matching algorithm
of Hopcroft and Karp~\cite{hk:match73} can be used to obtain a running time of
$O(m\sqrt n)$. In graphs where $d$ is a power of 2, the following elegant
idea, due to Gabow and Kariv~\cite{gk:edge1982}, leads to an algorithm with
$O(m)$ running time. First, compute an Euler tour of the graph (in time
$O(m)$) and then follow this tour in an arbitrary direction. Exactly half the
edges will go from left to right; these form a regular bipartite graph of
degree $d/2$. The total running time $T(m)$ thus follows the recurrence $T(m)
= O(m) + T(m/2)$ which yields $T(m) = O(m)$. Extending this idea to the
general case proved quite hard, and after a series of improvements (e.g. by
Cole and Hopcroft~\cite{ch:color82}, and then by Schrijver~\cite{s:color99} to
$O(md)$), Cole, Ost, and Schirra~\cite{cos:regular2001} gave an $O(m)$
algorithm for the case of general $d$. Their main interest was in edge
coloring of general bipartite graphs, where finding perfect matchings in
regular bipartite graphs is an important
subroutine.

Recently, Goel, Kapralov, and Khanna~\cite{gkk:soda09}, gave a sampling-based
algorithm that computes a perfect matching in $d$-regular bipartite graphs in
$O(\min\{m, \frac{n^{2.5}\log n}{d}\})$ expected time, an expression that is
bounded by $\tilde{O}(n^{1.75})$.  The algorithm of~\cite{gkk:soda09} uses
uniform sampling to reduce the number of edges in the input graph while
preserving a perfect matching, and then runs the Hopcroft-Karp algorithm on
the sampled graph. The authors of~\cite{gkk:soda09} also gave a lower bound of
$\tilde \Omega\left(\min\{nd, \frac{n^2}{d}\}\right)$ on the running time of
an algorithm that uses non-adaptive uniform sampling to reduce the number of
edges in the graph as the first step.  This lower bound was matched in
\cite{gkk:xv09}, where the authors use a two stage sampling scheme and a
specialized analysis of the runtime of the Hopcroft-Karp algorithm on the
sampled graph to obtain a runtime of $\tilde O\left(\min\{nd,
  \frac{n^2}{d}\}\right)$.

For sub-linear (in $m$) running time algorithms, the exact data model is
important. In this paper, as well as in the sub-linear running time algorithms
mentioned above, we assume that the graph is presented in the adjacency array
format, i.e., for each vertex, its $d$ neighbors are stored in an array. This
is the most natural input data structure for our problem. Our algorithms will
not make any ordering assumptions within an adjacency array.

Given a partial matching in an undirected graph, an augmenting path is a path
which starts and ends at an unmatched vertex, and alternately contains edges
that are outside and inside the partial matching. Many of the algorithms
mentioned above work by repeatedly finding augmenting paths.

\subsection{Our Results and Techniques}

Our main result is the following theorem.

\begin{theorem} \label{thm:main}
There exists an $O(n\log n)$ expected time algorithm for finding a perfect matching in a $d$-regular bipartite graph $G=(P, Q, E)$ given in adjacency array representation.
\end{theorem}

The algorithm is very simple: the matching is constructed by performing one
augmentation at a time, and new augmenting paths are found by performing a
random walk on a modified graph that encodes the current matching.  The random
walk approach may still be viewed as repeatedly drawing a uniform sample from
the adjacency array of some vertex $v$; however this vertex $v$ is now chosen
adaptively, thus allowing us to bypass the $\tilde \Omega\left(\min\{nd,
  \frac{n^2}{d}\}\right)$ lower bound on non-adaptive uniform sampling
established in~\cite{gkk:soda09}.  Somewhat surprisingly, we show that the
total time taken by these random augmentations can be bounded by $O(n \log n)$
in expectation, only slightly worse than the $\Omega(n)$ time needed to simply
output a perfect matching.  The proof involves analyzing the hitting time of
the sink node in the random walk. We also establish that randomization is
crucial to obtaining an $o(nd)$ time algorithm.

\begin{theorem}
\label{thm:detlb}
For any $1 \le d < n/8$, there exists a family of 
$d$-regular graphs on which any deterministic algorithm 
for finding a perfect matching requires $\Omega(nd)$ time.
\end{theorem}

\smallskip
Our techniques also 
extend to the problem of finding a perfect matching in the
support of a doubly-stochastic matrix, as well as to efficiently compute 
the Birkhoff-von-Neumann decomposition of a doubly stochastic matrix.

\begin{theorem} \label{thm:bvn}
  Given an $n \times n$ doubly-stochastic matrix $M$ with $m$ non-zero
  entries, one can find a perfect matching in the support of $M$ in
  $O(n\log^2 n)$ expected time with $O(m)$ preprocessing time.
\end{theorem}

In many applications of Birkhoff von Neumann decompositions (e.g. routing in network switches~\cite{clj:bvn_routing}), we need to find one perfect 
matching in a single iteration, and then update the weights of the matched edges. In such applications, each iteration can be implemented in $O(n \log^2 n)$ time (after initial
$O(m)$ preprocessing time), 
improving upon the previous best known bound of $O(mb)$ where $b$ is the bit precision.

\begin{corollary} \label{cor:bvn} For any $k \ge 1$, there exists an $O(m
  +kn\log^2 n)$ expected time algorithm for finding $k$ distinct matchings (if
  they exist) in the Birkhoff-von-Neumann decomposition of an $n\times n$
  doubly stochastic matrix with $m$ non-zero entries.
\end{corollary}

FInally, we note that an application of Yao's min-max theorem (see, for instance,~\cite{mr:random}) to
Theorem~\ref{thm:main} immediately yields the following corollary:
\smallskip
\begin{corollary}
  For any distribution on regular bipartite graphs with $2n$ nodes, there
  exists a deterministic algorithm that runs in average time $O(n\log n)$ on
  graphs drawn from this distribution.
\end{corollary}
A similar corollary also follows for doubly stochastic matrices.

\subsection{Organization}

Section~\ref{sec:matching_algo} gives the $O(n \log n)$ time algorithm to find
a perfect matching, and establishes Theorem~\ref{thm:main}.  Building on the
ideas developed in Section~\ref{sec:matching_algo}, we present in
Section~\ref{sec:bvn} algorithms for finding matchings in doubly-stochastic
matrices, and computing a Birkhoff-von-Neumann decomposition, establishing
Theorem~\ref{thm:bvn} and Corollary~\ref{cor:bvn}.  Finally, in
Section~\ref{sec:detlb}, we present an $\Omega(nd)$ lower bound for any
deterministic algorithm that finds a perfect matching in a $d$-regular graph.

\section{Matchings in $d$-Regular Bipartite Graphs}
\label{sec:matching_algo}

We start with a brief overview of our algorithm.  Let $G=(P, Q, E)$ denote the
input $d$-regular graph. Given any matching $M$ in $G$, the algorithm {\em
  implicitly} constructs a directed graph, called the {\em matching graph} of
$M$, with a source vertex $s$ and a sink vertex $t$. Any $s$ to $t$ path in
the graph $H$ defines an augmenting path in $G$ with respect to the matching
$M$.  The algorithm searches for an $s \leadsto t$ path using a random walk
from $s$, and once a path is found, the matching $M$ is augmented. We repeat
this process until we obtain a perfect matching. 

In what follows, we first describe the matching graph and prove a key lemma
bounding the expected running time of a random $s \leadsto t$ walk in $G$.
We then give a formal description of the algorithm, followed by an analysis of 
its running time.

\subsection{The Matching Graph}

Let $G=(P, Q, E)$ be a $d$-regular graph, and let $M$ be a partial matching
that leaves $2k$ nodes unmatched, for some integer $k$. The matching graph
corresponding to the matching $M$ is then defined to be the directed graph $H$
obtained by transforming $G$ as described below:

\begin{enumerate}
\item Orient edges of $G$ from $P$ to $Q$;
\item Contract each pair $(u, v)\in M$ into a supernode;
\item Add a vertex $s$ connected by $d$ parallel edges to each unmatched node
  in $P$, directed out of $s$;
\item Add a vertex $t$ connected by $d$ parallel edges to each unmatched node
  in $Q$, directed into $t$.
\end{enumerate}
The graph $H$ has $n+k+2$ nodes and $n(d-1)+k(2d+1)$ edges.  Note that for
every vertex $v\in H$, $v\neq s, t$ the in-degree of $v$ is equal to its
out-degree, the out-degree of $s$ equals $dk$, as is the in-degree of
$t$. Also, any path from $s$ to $t$ in $H$ gives an augmenting path in $G$
with respect to $M$. We now concentrate on finding a path from $s$ to $t$ in
$H$.

The core of our algorithm is a random walk on the graph $H$ starting at the
vertex $s$, in which an outgoing edge is chosen uniformly at random. The main
lemma in our analysis bounds the hitting time of vertex $t$ for a random walk
started at vertex $s$, when the number of unmatched vertices of $G$ is $2k$.

\begin{lemma} \label{lm:main} Given a $d$-regular bipartite graph $G$ and a
  partial matching $M$ that leaves $2k$ vertices unmatched, construct the
  graph $H$ as above. The expected number of steps before a random walk
  started at $s$ ends at $t$ is at most $2+\frac{n}{k}$.
\end{lemma}
\begin{proof}
Construct the graph $H^*$ by identifying $s$ with $t$ in $H$. Denote the vertex that corresponds to $s$ and $t$ by $s^*$.
Note that $H^*$ is a balanced directed graph, i.e. the out-degree of every vertex is equal to its in-degree. The out-degree of every vertex except $s^*$ is equal to $d-1$, while the out-degree of $s^*$ is $dk$. Consider a random walk in $H^*$ starting at $s^*$. The expected return time is equal to the inverse of the stationary measure of $s^*$. Since $H^*$ is a balanced directed graph, one has $\pi_i=\frac{\deg(i)}{\sum_{j\in V(H^*)}\deg(j)}$ for any vertex $i\in V(H^*)$, where $\pi_i$ is the stationary measure of vertex $i$. Thus, the expected time to return to $s^*$ is 
\begin{equation*}
\frac1{\pi_{s^*}}=\frac{\sum_{j\in V(H^*)}\deg(j)}{\deg(s^*)}=\frac{(n-k)(d-1)+2kd+kd}{kd}\leq 2+\frac{n}{k}.
\end{equation*}

It remains to note that running a random walk started at $s^*$ in $H^*$ until it returns to $s^*$ corresponds to running a random walk in $H$ starting at $s$ until it hits $t$. 
\end{proof}

\subsection{The Algorithm}

In what follows we shall use the subroutine TRUNCATED-WALK$(u, b)$, which we now define. The subroutine performs $b$ steps of a random walk on $H$ starting at a given vertex $u$:
\begin{enumerate}
\item If $u=t$, return SUCCESS
\item Set $v:=$SAMPLE-OUT-EDGE$(u)$. Set $b:=b-1$. If $b<0$ return FAIL, else return TRUNCATED-WALK$(v, b)$
\end{enumerate}

Here the function SAMPLE-OUT-EDGE$(u)$ returns the other endpoint of a
uniformly random outgoing edge of $u$. The implementation and runtime of
SAMPLE-OUT-EDGE depend on the representation of the graph. It is assumed in
Theorem \ref{thm:main} that the graph $G$ is represented in adjacency array
format, in which case SAMPLE-OUT-EDGE can be implemented to run in expected
constant time. In Theorem \ref{thm:bvn}, however, a preprocessing step will be
required to convert the matrix to an augmented binary search tree, in which
case SAMPLE-OUT-EDGE can be implemented to run in $O(\log n)$ time. We also
note that the graph $H$ is never constructed explicitly since SAMPLE-OUT-EDGE
can be implemented using the original graph.

The following lemma is immediate:
\begin{lemma}
An augmenting path of length at most $b$ can be obtained from the sequence of steps taken by a successful run of TRUNCATED-WALK$(s, b)$ by removing loops.
\end{lemma}

We can now state our algorithm: 
\begin{description}
\item[Input:] A $d$-regular bipartite graph $G=(P, Q, E)$ in adjacency array
  format.
\item[Output:] A perfect matching of $G$.
\item[1.] Set $j:=0$, $M_0:=\emptyset$.
\item[2.] Repeatedly run TRUNCATED-WALK$(s, b_j)$ where $b_j =
  2 \left( 2 +\frac{n}{n-j} \right)$ on the matching graph $H$, implicitly defined by $M_j$,
  until a successful run is obtained.
\item[3.] Denote the augmenting path obtained by removing possible loops from
  the sequence of steps taken by the walk by $p$. Set $M_{j+1}:=M_j \Delta p$.
\item[4.] Set $j:=j+1$ and go to step 2.
\end{description}

\subsection{Time Complexity}

We now analyze the running time of our algorithm, and prove Theorem \ref{thm:main}.

\begin{proofof}{Theorem \ref{thm:main}}
  We show that the algorithm above takes time $O(n\log n)$ whp. First note
  that by Lemma \ref{lm:main} and Markov's inequality TRUNCATED-WALK$(s, b)$
  succeeds with probability at least $1/2$.  Let $X_j$ denote the time taken by the
  $j$-th augmentation. Let $Y_j$ be independent exponentially
  distributed with mean $\mu_j := \frac{b_j}{ \ln 2}$. Note that
\begin{equation*}
\begin{split}
  \prob[X_j\geq qb_j]\leq
  2^{-q}=\exp\left[-\frac{qb_j\ln 2}{b_j}\right]=\prob[Y_j\geq q b_j]\\
\end{split}
\end{equation*}
for all $q> 1$, so 
\begin{equation} \label{y-dominates}
\prob[X_j\geq x]\leq \prob[Y_j\geq x]
\end{equation}
for all $x>b_j$. We now prove that $Y:=\sum_{0\leq j\leq n-1} Y_j\leq cn\log
n$ w.h.p. for a suitably large positive constant $c$. Denote
$\mu:=\expect[Y]$. 
By Markov's inequality, for any $t,\delta > 0$
\begin{equation*}
\prob[Y \ge (1+\delta)\mu] \le \frac{\expect[e^{tY}]}{e^{t(1+\delta)\mu}}.
\end{equation*}
Also, for any $j$, and for $t<1/\mu_j$, we have
\begin{equation*}
  \expect[e^{tY_j}]=\frac{1}{\mu_j}\int_{0}^{\infty} e^{tx}
  e^{-x/\mu_j}dx=\frac{1}{1-t\mu_j}.
\end{equation*}
The two expressions above, along with the fact that the $Y_j$'s are
independent, combine to give:
\begin{equation}
  \label{eq:chernoff}
  \prob[Y \ge (1+\delta)\mu] \le \frac{e^{-t(1+\delta)\mu}}{\prod_{j=0}^{n-1} (1-t\mu_j)}.
\end{equation}
Observe that $\mu_{n-1}$ is the largest of the $\mu_j$'s. Assume that $t =
\frac{1}{2\mu_{n-1}}$, which implies that $(1-t\mu_j) \ge e^{-t\mu_j \ln 4}$.
    Plugging this into equation~\ref{eq:chernoff}, we get:
\begin{equation}
  \label{eq:chernoff2}
  \prob[Y \ge (1+\delta)\mu] \le e^{-(1+\delta-\ln 4)\mu/(2\mu_{n-1})}.
\end{equation}
Further observe that $\mu = 2n/\ln 2+(\mu_{n-1}-2/\ln 2) H(n)\geq \mu_{n-1} H(n)$, where $H(n) := 1 + 1/2 + 1/3 +
\ldots +1/n$ is the $n$-th Harmonic number. Since $H(n) \ge \ln n$, we get our
high probability result:
\begin{equation}
  \label{eq:chernoff2}
  \prob[Y \ge (1+\delta)\mu] \le n^{-(1+\delta-\ln 4)/2}.
\end{equation}
Since $\mu = O(n \log n)$, this completes the proof of Theorem~\ref{thm:main}.
\end{proofof}

\begin{remark}
  The truncation step is crucial for a high probability result. However, if
  one were interested only in obtaining a small expected running time, then
  the untruncated walk would result in $O(n\log n)$ time directly. In fact,
  the expected total number of steps in all the untruncated random walks is at
  most $n + nH(n)$, directly by Lemma~\ref{lm:main}.
\end{remark}
\begin{remark}
  The algorithm above can be used to obtain a simple algorithm for
  edge-coloring bipartite graphs with maximum degree $d$ in time $O(m\log^2
  n)$ (slightly worse than the $O(m\log d)$ dependence obtained in
  \cite{cos:regular2001}). In the first step one reduces the problem to that
  on a regular graph with $O(m)$ edges as described in \cite{cos:regular2001}.
  The lists of neighbors of every vertex of the graph can then be arranged in
  a data structure that supports fast uniform sampling and deletions (for
  example, one can use a skip-list, resulting in $O(\log n)$ runtime for
  SAMPLE-OUT-EDGE). It remains to find matchings repeatedly, taking $O(n\log^2
  n)$ time per matching (the extra $\log n$ factor comes from the runtime of
  SAMPLE-OUT-EDGE, assuming an implementation using skip-lists). This takes
  $O(nd\log^2 n)=O(m\log^2 n)$ time (note that $n$ is the number of vertices
  in the regular graph, which could be different from the number of vertices
  in the original graph).
\end{remark}

\section{Matchings in Doubly-Stochastic Matrices}
\label{sec:bvn}

We now apply techniques of the previous section to the problem of finding a
perfect matching in the support of an $n\times n$ doubly stochastic matrix
$\M$ with $m$ non-zero entries. A doubly-stochastic matrix can be viewed as a
regular graph possibly with parallel edges, and we can thus use the same
algorithm and analysis as above, provided that SAMPLE-OUT-EDGE can be
implemented efficiently. We start by describing a simple data structure for
implementing SAMPLE-OUT-EDGE. For each vertex $v$, we store all the outgoing
edges from $v$ in a balanced binary search tree, augmented so that each node
in the search tree also stores the weight of all the edges in its
subtree. Inserts into, deletes from, and random samples from this augmented
tree all take time $O(\log n)$~\cite{clrs}, giving a running time of
$O(n\log^2 n)$ for finding a matching in the support of the doubly stochastic
matrix.

Superficially, it might seem that initializing the balanced binary search
trees for each vertex will take total time $\Theta(m\log n)$. However, note that
there is no natural ordering on the outgoing edges from a vertex, and we can
simply superimpose the initial balanced search tree for a vertex on the
adjacency array for that vertex, assuming that the underlying keys are in
accordance with the (arbitrary) order in which the edges occur in the
adjacency array.

The complete Birkhoff-von Neumann decomposition can be computed by subtracting
an appropriately weighted matching matrix from $\M$ every time a matching is
found, thus decreasing the number of nonzero entries of $\M$. Note that the
augmented binary search tree can be maintained in $O(\log n)$ time per
deletion. This yields the algorithm claimed in Corollary \ref{cor:bvn}.

\section{An $\Omega(nd)$ Lower Bound for Deterministic Algorithms}
\label{sec:detlb}
\newcommand{\A}{{\mathcal A}}
\newcommand{\Q}{{\mathcal Q}}
\newcommand{\adj}{{\sf Adj}}
\renewcommand{\deg}{{\sf deg}}
\newcommand{\G}{{\mathcal G}}

In this section, we will prove Theorem~\ref{thm:detlb}. We will show that for
any positive integer $d$, any deterministic algorithm to find a perfect
matching in a $d$-regular bipartite graph requires $\Omega(nd)$ probes, even
in the adjacency array representation, where the ordering of edges in an array
is decided by an adversary.  Specifically, for any positive integer $d$, we
construct a family $\G(d)$ of $d$-regular bipartite graphs with $O(d)$
vertices each that we refer to as {\em canonical} graphs.  A canonical
bipartite graph $G(P \cup \{t\},Q \cup \{s\},E) \in \G(d)$ is defined as
follows.  The vertex set $P = P_1 \cup P_2$ and $Q = Q_1 \cup Q_2$ where
$|P_i| = |Q_i| = 2d$ for $i \in \{ 1,2 \}$. The vertex $s$ is connected to an
arbitrary set of $d$ distinct vertices in $P_1$ while the vertex $t$ is
connected to an arbitrary set of $d$ distinct vertices in $Q_2$. In addition,
$G$ contains a {\em perfect matching} $M'$ of size $d$ that connects a subset
$Q'_1 \subseteq Q_1$ to a subset $P'_2 \subseteq P_2$, where $|Q'_1| =
|P'_2|=d$.  The remaining edges in $E$ connect vertices in $P_i$ to $Q_i$ for
$i \in \{ 1,2 \}$ so as to satisfy the property that the degree of each vertex
in $G$ is exactly $d$.  It suffices to show an $\Omega(d^2)$ lower bound for
graphs drawn from $\G(d)$ since we can take $\Theta(n/d)$ disjoint copies of
canonical graphs to create a $d$-regular graph on $n$ vertices.

\medskip
\noindent
{\bf Overview:}
Let $A$ be a deterministic algorithm for finding a perfect matching in graphs
drawn from $\G(d)$. We will analyze a game between the algorithm $A$ and an
adaptive adversary $\A$ whose goal is to maximize the number of edges that $A$
needs to examine in order to find a perfect matching. In order to find a
perfect matching, the algorithm $A$ must find an edge in $M'$, since $s$ must
be matched to a vertex in $P_1$, and thus in turn, some vertex in $Q_1$ must
be matched to a vertex in $P_2$.  We will show that the adversary $\A$ can
always force $A$ to examine $\Omega(d^2)$ edges in $G$ before revealing an
edge in $M'$.  The specific graph $G \in \G(d)$ presented to the algorithm
depends on the queries made by the algorithm $A$.  The adversary adaptively
answers these queries while maintaining at all times the invariant that the
partially revealed graph is a subgraph of some graph $G \in \G(d)$.  The cost
of the algorithm is the number of edge locations probed by it before $\A$
reveals an edge in $M'$ to $A$.

In what follows, we assume that the adversary reveals $s, t$ and the partition
of remaining vertices into $P_i, Q_i$ for $1 \le i \le 2$, along with all
edges from $s$ to $P_1$ and all edges from $t$ to $Q_2$, to the algorithm at
the beginning. The algorithm pays no cost for this step.

\medskip
\noindent {\bf Queries:} Whenever the algorithm $A$ probes a new location in
the adjacency array of some vertex $u \in P \cup Q$, we will equivalently view
$A$ as making a query $\Q(u)$ to the adversary $\A$, in response to which the
adversary outputs a vertex $v$ that had not been so far revealed as being
adjacent to $u$.

\medskip
\noindent
{\bf Subgraphs consistent with canonical graphs:}
Given a bipartite graph $G'(P \cup \{t\},Q \cup \{s\},E')$, we say that a
vertex $u \in P \cup Q$ is {\em free} if its degree in $G'$ is strictly
smaller than $d$. The lemma below identifies sufficient conditions for a graph
to be a subgraph of some canonical graph in $\G(d)$.

\begin{lemma}
\label{lem:lb1}
Let $G_r(P \cup \{t\},Q \cup \{s\},E_r)$ be any bipartite graph such that 

\begin{itemize}
\item[(a)]
the vertex $s$ is connected to $d$ distinct vertices
in $P_1$ and the vertex $t$ is connected to $d$ distinct vertices in $Q_2$, 

\item[(b)]
all other edges in $G_r$ connect a vertex in $P_i$ to a vertex in $Q_i$ for
some $i \in \{ 1,2 \}$, 

\item[(c)]
degree of each vertex in $G_r$ is at most $d$, and

\item[(d)]
there exist at least $(d+1)$ free vertices each in both $Q_1$ and $P_2$.
\end{itemize}

Also, let $u, v$ be a pair of vertices such that $u \in P_i$ and $v \in Q_i$
for some $i \in \{ 1,2 \}$, and $(u,v) \not \in E_r$.
Then there exists a canonical graph $G(P \cup \{t\},Q \cup \{s\},E) \in \G(d)$ 
such that that $E_r \cup {(u,v)} \subseteq E$ iff both $u$ and $v$ have degree 
strictly less than $d$ in $G_r$.
\end{lemma}
\begin{proof}
If either $u$ or $v$ has degree $d$, then clearly addition of the edge $(u,v)$ 
would violate the $d$-regularity condition satisfied by all graphs in $\G(d)$. 

If both $u$ and $v$ have degree strictly less than $d$, then let $G'(P \cup
\{t\},Q \cup \{s\},E')$ be the graph obtained by adding edge $(u,v)$ to $G_r$,
that is, $E ' = E_r \cup \{(u,v)\}$. The degree of each vertex in $G'$ remains
bounded by $d$. We now show how $G'$ can be extended to a $d$-regular
canonical graph.

We first add to $G'$ a perfect matching $M'$ of size $d$ connecting an
arbitrary set of $d$ {\em free} vertices in $Q_1$ to an arbitrary set of $d$
{\em free} vertices in $P_2$.  This is feasible since by assumption, the graph
$G_r$ had at least $(d+1)$ free vertices each in both $Q_1$ and $P_2$, and
thus addition of the edge $(u,v)$ still leaves at least $d$ free vertices each
in $Q_1$ and $P_2$.  In the resulting graph, since the total degree of all
vertices in $P_i$ is same as the total degree of all vertices in $Q_i$, we can
now repeatedly pair together a vertex of degree less than $d$ in $P_i$ with a
vertex of degree less than $d$ in $Q_i$ until degree of each vertex becomes
exactly $d$, for $i \in \{1,2\}$. The resulting graph satisfies all properties
of a canonical graph in the family $\G(d)$.
\end{proof}

\medskip
\noindent
{\bf Adversary strategy:}
For each vertex $u \in P \cup \{t\},Q \cup \{s\}$, the adversary $\A$
maintains a list $N(u)$ of vertices adjacent to $w$ that have been so far
revealed to the algorithm $A$.  Wlog we can assume that the algorithm $A$
never queries a vertex $u$ for which $|N(u)| = d$. At any step of the game,
we denote by $G_r$ the graph
formed by the edges revealed thus far.  We say the game is in {\em evasive}
mode if the graph $G_r$ satisfies the condition $(a)$ through $(d)$ of
Lemma~\ref{lem:lb1}, and is in {\em non-evasive} mode otherwise. Note that the
game always starts in the evasive mode, and then switches to non-evasive mode.

When the game is in the evasive mode, in response to
a query $\Q(u)$ by $A$ for some free vertex $u \in P_i$ ($i \in \{ 1,2 \}$), 
$\A$ returns an arbitrary free vertex $v \in Q_i$ such that $v \not\in  N(u)$.
The adversary then adds $v$ to $N(u)$ and $u$ to $N(v)$. 
Similarly, when $A$ asks a query $\Q(u)$ for some free vertex $u \in Q_i$ ($i \in \{ 1,2 \}$), 
$\A$ returns an arbitrary free
vertex $v \in P_i$ such that $v \not\in  N(u)$.
It then adds $v$ to $N(u)$ and $u$ to $N(v)$ as above.

As the game transitions from evasive to non-evasive mode, Lemma~\ref{lem:lb1} ensures
existence of a canonical graph $G \in \G(d)$ that contains the graph revealed by the
adversary thus far as a subgraph.  The adversary answers all subsequent
queries by $A$ in a manner that is consistent with the edges of $G$.

The lemma below shows that the simple adversary strategy above forces 
$d^2$ queries before the evasive mode terminates.

\begin{lemma}
The algorithm makes at least $d^2$ queries before the game enters non-evasive mode. 
\end{lemma}
\begin{proof}
The adversary strategy ensures that conditions $(a)$ through $(c)$ in
Lemma~\ref{lem:lb1} are maintained at all times as long as the game is
in the evasive mode. So we consider the first time that condition $(d)$ is violated.
Since each query answered by the adversary in the evasive mode contributes $1$ to the 
degree of exactly one vertex in $Q_1 \cup P_2$,  
$\A$ must answer at least $d^2$ queries before the number of free vertices 
falls below $(d+1)$ in either $Q_1$ or $P_2$. The lemma follows.
\end{proof}

\medskip

Since $A$ can not discover an edge in $M'$ until the game enters the
non-evasive mode, we obtain the desired lower bound of $\Omega(d^2)$.

\pdfbookmark[1]{\refname}{My\refname} \bibliographystyle{plain}

\end{document}